\documentclass[11pt,dvipsnames]{article}
\usepackage[letterpaper,margin=1in]{geometry}
\usepackage{enumitem}
\usepackage[noblocks]{authblk}
\setitemize{itemsep=0pt}

\usepackage[T1]{fontenc}
\usepackage[english]{babel}
\usepackage[utf8]{inputenc}
\usepackage{amsmath,amssymb,amsfonts,mathrsfs,amsthm}
\usepackage{graphicx}
\usepackage{pdfpages}

\usepackage{multirow}
\usepackage{varioref}
\usepackage{datetime}
\usepackage{mathtools}
\usepackage{ifthen}
\usepackage{stmaryrd}
\usepackage[h]{esvect}
\usepackage{array}
\usepackage{listings}
\usepackage{booktabs}
\usepackage{algorithm}
\usepackage[noend]{algpseudocode}
\usepackage{makecell}
\usepackage{pdflscape}
\usepackage{subcaption}
\usepackage[linkcolor=black,citecolor=black,filecolor=black]{hyperref}
\usepackage[capitalise]{cleveref}
\usepackage{complexity}

\numberwithin{equation}{section}

\newtheorem{theorem}{Theorem}

\newtheorem{corollary}[theorem]{Corollary}
\newtheorem{lemma}[theorem]{Lemma}

\newtheorem{observation}[theorem]{Observation}

\theoremstyle{definition}
\newtheorem{definition}[theorem]{Definition}

\crefname{claim}{claim}{claims}
\crefname{observation}{observation}{observations}

\renewcommand{\epsilon}{\ensuremath\varepsilon}

\renewcommand{\phi}{\ensuremath{\varphi}}

\renewcommand{\epsilon}{\ensuremath{\varepsilon}}

\renewcommand{\theta}{\ensuremath{\vartheta}}

\newcommand{\matousek}{Matou\v{s}ek}

\newcommand{\szabo}{Sz\'abo}

\usepackage{thm-restate}

\newcommand{\xor}{\oplus}
\newcommand{\bitwiseand}{\odot}
\DeclareRobustCommand\zerovec{\mathbf{0}}
\DeclareRobustCommand\onesvec{\mathbf{1}}
\DeclareMathOperator{\Span}{span}

\title{Realizability Makes a Difference: A Complexity Gap for Sink-Finding in USOs}

\let\anonymous\undefined 

\ifdefined\anonymous
    \author{Redacted}
    \affil{Affiliation\\ \texttt{email@adre.ss}}
    \usepackage[]{lineno}
    \linenumbers
\else
    \author{Simon~Weber}
    \affil{Department of Computer Science\\ ETH Zürich\\ \href{mailto:simon.weber@inf.ethz.ch}{\texttt{simon.weber@inf.ethz.ch}}}
    \author{Joel~Widmer}
    \affil{Department of Mathematics\\ ETH Zürich\\ \href{mailto:joewidme@student.ethz.ch}{\texttt{joewidme@student.ethz.ch}}}
\fi

\date{}

\begin{document}

\maketitle
\thispagestyle{empty}
\setcounter{page}{0}
\begin{abstract}
    Algorithms for finding the sink in Unique Sink Orientations (USOs) of the hypercube can be used to solve many algebraic and geometric problems, most importantly including the {P\nobreakdash-Matrix} Linear Complementarity Problem and Linear Programming.
    The \emph{realizable} USOs are those that arise from the reductions of these problems to the USO sink-finding problem. Finding the sink of realizable USOs is thus highly practically relevant, yet it is unknown whether realizability can be exploited algorithmically to find the sink more quickly.
    However, all (non-trivial) known unconditional lower bounds for sink-finding make use of USOs that are provably not realizable. This indicates that the sink-finding problem might indeed be strictly easier on realizable USOs.
    
    In this paper we show that this is true for a subclass of all USOs.
    We consider the class of \matousek{}-type USOs, which are a translation of \matousek{}'s LP-type problems into the language of USOs.
    We show a query complexity gap between sink-finding in \emph{all},
    and sink-finding in only the \emph{realizable} $n$-dimensional \matousek{}-type USOs.
    We provide concrete deterministic algorithms and lower bounds for both cases, and show that in the realizable case~$O(\log^2n)$ vertex evaluation queries suffice, while in general exactly~$n$ queries are needed.
    The \matousek{}-type USOs are the first USO class found to admit such a gap.
\end{abstract}

\ifdefined\anonymous
\else
{\footnotesize
\vfill
\textbf{Acknowledgments.}
This research is supported by the Swiss National Science Foundation under project no.~204320. We thank Bernd Gärtner for his very helpful insights and feedback.
}
\fi

\newpage

\section{Introduction}
A Unique Sink Orientation (USO) is an orientation of the hypercube graph with the property that each face has a unique sink. The most studied algorithmic problem related to USOs is that of finding the sink: An algorithm has access to a \emph{vertex evaluation oracle}, which can be queried with a vertex and returns the orientation of all incident edges. The task is to determine the unique sink of the USO using as few such vertex evaluation queries as possible.

Progress on this problem has stalled for a long time. Since \szabo{} and Welzl introduced USOs in~2001~\cite{szabo2001usos}, their deterministic and randomized algorithms --- both requiring an exponential number of queries in terms of the hypercube dimension --- are still the best known in the general case. Only for special cases, such as acyclic USOs, better algorithms are known~\cite{gaertner2002simplex}.

\paragraph{Realizability.} 
The study of USOs and the sink-finding problem was originally motivated by a reduction of the P-Matrix Linear Complementarity Problem (P-LCP) to sink-finding in USOs~\cite{stickney1978digraph}.
Many widely studied optimization problems have since been shown to be reducible either to the P-LCP or to sink-finding in USOs directly, the most notable example being Linear Programming~(LP)~\cite{gaertner2006lpuso}, but also Convex Quadratic Programming~\cite{schurr2004phd} or the Smallest Enclosing Ball problem~\cite{gaertner2001enforcing}. To make progress on this wide array of problems, we do not need to find an algorithm which can find the sink quickly in \emph{all} USOs, but only in the USOs which can arise from these reductions. We call the USOs generated by the reduction from P-LCP \emph{realizable}. The realizable USOs encompass the USOs arising from all the aforementioned reductions.

The number of $n$-dimensional realizable USOs is much smaller than the total number of USOs, namely $2^{\Theta(n^3)}$ in contrast to $2^{\Theta(2^n\log n)}$~\cite{foniok2014counting}. Furthermore, a simple combinatorial property, the Holt-Klee condition, is known to hold for all realizable USOs~\cite{gaertner2008grids,holt1999klee}. Lastly, all (non-trivial) known unconditional lower bound constructions, including the best-known lower bound for deterministic algorithms of $\Omega(n^2/\log n)$~\cite{schurr2004quadraticbound}, may produce USOs which fail the Holt-Klee condition and are therefore not realizable. These three facts indicate that it could be possible to algorithmically exploit the features of realizable USOs to beat the algorithms of \szabo{} and Welzl on this important class of USOs.

For P-LCP and certain cases of LP, the fastest deterministic combinatorial algorithms are already today based on sink-finding in USOs~\cite{fearnley2018ueopl,gaertner2006lpuso}. Improved sink-finding algorithms for realizable USOs would directly translate to advances in P-LCP and LP algorithms. In particular, a sink-finding algorithm using polynomially many vertex evaluations on any realizable USO would imply the existence of a strongly polynomial-time algorithm for LP, answering a question from Smale's list of \emph{mathematical problems for the next century}~\cite{smale1988problems}.

\paragraph{\matousek{}(-type) USOs.} In 1994, \matousek{} introduced a family of LP-type problems to show a superpolynomial lower bound on the runtime of the Sharir-Welzl algorithm~\cite{matousek1994lowerbound}. This result was later translated into the framework of USOs, where the corresponding class of USOs is now called the \emph{\matousek{}~USOs} and the lower bound translates to the query complexity of the \emph{\textsc{Random~Facet}} algorithm~\cite{gaertner2002simplex}. In the same paper it was also shown that the sink of all \matousek{} USOs that fulfill the Holt\nobreakdash-Klee property (thus including all realizable ones) is found by the \textsc{Random Facet} algorithm in a quadratic number of queries. Therefore, \textsc{Random Facet} is strictly faster on the realizable \matousek{} USOs than on all \matousek{} USOs. In this paper we aim to provide a similar result for the query complexity of the problem itself, instead of a concrete algorithm.

The \matousek{} USOs all have the sink at the same vertex. This does not pose a problem when analyzing a fixed algorithm  which does not exploit this fact (e.g., \textsc{Random Facet}), but it does not allow us to derive algorithm-independent lower bounds. To circumvent this issue, we consider the class of \emph{\matousek{}-type USOs}, which simply contains all orientations isomorphic to classical \matousek{} USOs.

It was recently discovered that all \matousek{}-type USOs fulfilling the Holt-Klee property are also realizable~\cite{weber2021matousek}, showing that the Holt-Klee property is not only necessary but also sufficient for realizability in \matousek{}-type USOs. The proof of this result employed a novel view on \matousek{}-type USOs, describing such a USO completely by (i) the location of the sink, and (ii) a directed graph with the $n$ dimensions of the hypercube as vertices, called the \emph{dimension influence graph}. We make heavy use of this view in the proofs of our results.

\paragraph{Results.}
We show a query complexity gap between sink-finding on the realizable and sink-finding on all \matousek{}-type USOs.
We achieve this by proving the following two main theorems:

\begin{restatable}{theorem}{lowerbound}
\label{thm:lowerbound} For every deterministic sink-finding algorithm $\mathcal{A}$ and any $n\geq 2$, there exists some $n$-dimensional \matousek{}-type USO on which $\mathcal{A}$ requires at least $n$ vertex evaluations to find the sink.
\end{restatable}
\begin{restatable}{theorem}{sublinear}
\label{thm:sublinear} There exists a deterministic algorithm finding the sink of any $n$-dimensional \emph{realizable} \matousek{}-type USO using $O(\log^2 n)$ vertex evaluations in the worst case.
\end{restatable}
\noindent
In addition, we show that the result about general \matousek{}-type USOs is tight. For the realizable case, we provide a simple lower bound of $\Omega(\log n)$ vertex evaluations.

\paragraph{Discussion.}
The \matousek{}-type USOs form the first known USO class admitting such a complexity gap. We hope that this result motivates further research into tailored algorithms using the property of realizability for larger, more relevant classes of USOs.

Note that an artificial class of USOs exhibiting such a complexity gap could easily be constructed by combining a set $R$ of easy-to-solve realizable USOs with a set $N$ of difficult-to-solve non-realizable USOs. For $R$, one could take any set of realizable USOs which all have the same vertex as their sink. An algorithm to find the sink of USOs in $R$ could then always output this vertex without needing to perform any vertex evaluations. For the set $N$, one could take the set of USOs constructed in the lower bound of Schurr and \szabo{}~\cite{schurr2004phd}, and change each USO such that it becomes non-realizable. This can be achieved without destroying the lower bound. The resulting class $R\cup N$ would then also exhibit a complexity gap.

The \matousek{}-type USOs are far away from being such an artificially constructed class of USOs. First off, they are well-studied due to their significance in proving the lower bound for the \textsc{Random Facet} algorithm~\cite{matousek1994lowerbound,gaertner2002simplex}. Second, they can be considered a natural choice for proving unconditional lower bounds for realizable USOs: All known unconditional lower bounds on general USOs use decomposable USOs~\cite{schurr2004quadraticbound}, and the realizable \matousek{}-type USOs are the only known class of realizable decomposable USOs.

Even on a natural USO class, a complexity gap could be trivial, for example if the class contains no (or only very few) realizable USOs. This is also \emph{not} the case for the \matousek{}-type USOs, as there are $2^{\Theta(n\log n)}$ realizable $n$-dimensional \matousek{}-type USOs, while the overall number of \matousek{}-type USOs is $2^{\Theta(n^2)}$. This is a much larger realizable fraction than one observes on the set of all USOs.

The \matousek{}-type USOs neatly connect to the D-cubes, which are a subset of realizable USOs including those arising from the reduction of LP to sink-finding. The \emph{dimension influence graph} encoding the structure of a \matousek{}-type USO can be viewed as a global version of the \emph{L-graphs} used in a recent necessary condition found to hold for all D-cubes~\cite{gao2020dcubes}. The techniques developed in this work to find the sink in the more rigid \matousek{}-type USOs might even be useful in developing algorithms for D-cubes. While all \matousek{}-type USOs fulfill the necessary condition for D-cubes, it remains open whether the realizable \matousek{}-type USOs are in fact D-cubes.

\paragraph{Proof Techniques.}
For both the lower and the upper bounds, we need to view the \matousek{}-type USOs by their dimension influence graphs. We first show an equivalence of finding a certain subset of the vertices in the dimension influence graph to finding the sink in the USO itself. Considering the adjacency matrix $M$ of the dimension influence graph, finding the desired subset of vertices can be viewed as solving a linear system of equations $Mx=y$ over $GF(2)$, where $M$ is only given by a matrix-vector product oracle, answering queries $q\in\{0,1\}^n$ with $Mq$.

For the lower bounds, we provide adversarial constructions to adapt the adjacency matrix of the dimension influence graph to the queries of the algorithm.
Starting with the identity matrix, $M$ is changed after some or all queries to ensure that the algorithm is not able to deduce the solution~$x$. The main technical difficulties are to simultaneously ensure three conditions on a changing $M$: (i)~that the previously given replies are consistent with the new matrix, (ii) that the algorithm still cannot deduce $x$ after the change, and (iii)~that the matrix describes a legal dimension influence graph of a (realizable) \matousek{}-type~USO.

For the upper bound in the realizable case, our algorithm does not find the sink of the USO directly, but instead recovers the whole dimension influence graph and thus the orientation of all USO edges. It can then compute the location of the sink without needing any more queries. The reply to each query contains only $n$ bits of information, and the adjacency matrix~$M$ consists of $n^2$ bits. Any sublinear algorithm discovering the whole adjacency matrix must therefore leverage some additional structure of the graph. We use a previous result stating that the dimension influence graph of every realizable \matousek{}-type USOs is the reflexive transitive closure of a branching~\cite{weber2021matousek}. Thanks to this rigid structure, we can split the problem into multiple subproblems, and we can combine queries used to solve different subproblems into a single query. This allows us to make progress on many subproblems in ``parallel'', leading to fewer queries needed overall.

\paragraph{Paper Overview.} In \Cref{sec:preliminaries}, we lay out the necessary notations and definitions for the paper and prove the equivalence of our considered problem variants. In \Cref{sec:generalcase}, we give matching lower and upper bounds of $n$ vertex evaluations for the case of general \matousek{}-type USOs. In \Cref{sec:realizablecase}, we show the improvements which can be made in the realizable case. Finally, we discuss remaining open questions in \Cref{sec:conclusion}.

\section{Preliminaries}\label{sec:preliminaries}
We begin with some basic notation. All vectors and matrices in this paper are defined over the field $GF(2)$. We write $\bitwiseand$ and $\xor$ for bit-wise multiplication (``and'') and addition (``xor'') in $GF(2)$. By $\zerovec$ (or $\onesvec$) we denote the all-zero (or all-ones) $n$-dimensional vector. By $e_i$ we denote the $i$-th standard basis vector. $I$ denotes the $n$-dimensional identity matrix. For a natural number $x$, we write $Bin(x)_i\in\{0,1\}$ for the $i$-th least significant bit of the binary representation of $x$, such that $\sum_{i=0}^{\infty}Bin(x)_i\cdot 2^i=x.$

\subsection{Orientations and USOs}
The $n$-dimensional hypercube is an undirected graph $(V,E)$ consisting of the vertex set $V=\{0,1\}^n$, where two vertices are connected by an edge if they differ in exactly one coordinate. An orientation of the hypercube assigns a direction to each of the $n2^{n-1}$ edges.

\begin{definition}[Hypercube Orientation]
An \emph{orientation} $o$ of the $n$-dimensional hypercube is described by a function $o:\{0,1\}^n\rightarrow\{0,1\}^n$ assigning each vertex its \emph{outmap}. An edge $(v,v\xor e_i)$ is directed away from $v$ if $o(v)_i=1$. To ensure consistent orientation of all edges, $o$ has to fulfill $o(v)_i \not= o(v\xor e_i)_i$ for all $v\in V$ and $i\in [n]$.
\end{definition}

\begin{definition}[Unique Sink Orientation]
A \emph{Unique Sink Orientation (USO)} is an orientation of the hypercube, such that for each non-empty face $F$ of the hypercube, the subgraph induced by $F$ has a unique sink, i.e., a unique vertex which has no outgoing edges.
\end{definition}

A sink-finding algorithm has access to the orientation function $o$ as an oracle. Given a vertex $v\in\{0,1\}^n$, the oracle returns its outmap $o(v)$. We are only interested in the number of such \emph{vertex evaluation} queries made, and allow the algorithm to perform an unbounded amount of additional computation and put no restriction on the allowed memory consumption.

\begin{definition}[Realizability]
A USO $o$ is \emph{realizable}, if there is a non-degenerate P-Matrix Linear Complementarity Problem (P-LCP) instance~$(M,q)$ such that the reduction of this instance to USO sink-finding produces $o$.
\end{definition}

We omit the formal definitions of the P-LCP and this reduction, as they are not needed to derive or understand our results. This reduction first appeared in the seminal paper of Stickney and Watson~\cite{stickney1978digraph} in 1978.  As their orientations were only named USOs much later~\cite{szabo2001usos}, we point the interested reader to the comprehensive PhD thesis of Klaus~\cite{klaus2012phd}, which uses more modern language.

\subsection{(Realizable) \matousek{}-type USOs}
A \matousek{} USO, as defined by Gärtner~\cite{gaertner2002simplex}, is an orientation $o$ characterized by an invertible, upper-triangular matrix $A\in \{0,1\}^{n\times n}$ (thus all diagonal entries of $A$ are $1$). The matrix defines the orientation $o(v) = Av$. It is easy to see that each principal submatrix of $A$ must also be invertible. This implies that there must be a unique sink in each face of the hypercube. In particular, the whole hypercube has a unique sink at the vertex $\zerovec$.

To eliminate this commonality among \matousek{} USOs, we define the \matousek{}-type USOs, which are all orientations isomorphic to a \matousek{} USO. Isomorphisms on the hypercube allow for mirroring of any subset of dimensions, and for relabeling the dimensions.

\begin{definition}[\matousek{}-type USO]\label{def:matousekuso}
A \emph{\matousek{}-type USO} is an orientation $o$, with
\[\forall v\in\{0,1\}^n:\;\;o(v)=M(v\xor s), \text{ where } M:=PAP^T\]
for some permutation matrix $P$, an invertible, upper-triangular matrix $A\in \{0,1\}^{n\times n}$, and the desired location of the sink $s\in\{0,1\}^n$.
\end{definition}
We can view the matrix $M$ as the adjacency matrix of a directed graph $G=([n],E_M)$ on the dimensions $[n]$, where $(i,j)\in E_M$ if $M_{j,i}=1$. As $A$ is invertible and upper-triangular, and as $M$ is equal to $A$ with rows and columns permuted in the same way, $G$ is an acyclic graph with additional loop edges $(i,i)$ at every vertex $i$. We also say that $G$ is the \emph{reflexive closure} of an acyclic graph.

We call this graph $G$ the \emph{dimension influence graph} of the \matousek{}-type USO. The name is motivated by the following observation.
\begin{observation}
Let $\lambda,\phi\in [n]$ be two distinct dimensions of a \matousek{}-type USO $o$.
For any vertex $v\in\{0,1\}^n$, it holds that
\[o(v)_\lambda\not=o(v\xor e_\phi)_\lambda \Longleftrightarrow M_{\lambda,\phi}=1\Longleftrightarrow (\phi, \lambda)\in E_M.\]
\end{observation}

Intuitively, this means that in a \matousek{}-type USO, any $2$-dimensional face spanned by the same two dimensions $\lambda$ and $\phi$ has the same structure: ``Walking along'' an edge in dimension $\phi$ either always changes the direction of the adjacent $\lambda$-edge, or never (see \Cref{fig:dimensioninfluence}). If it always changes, we say that $\phi$ \emph{influences} $\lambda$, and there is an edge from $\phi$ to $\lambda$ in the dimension influence graph. See \Cref{fig:forbidden} for two example dimension influence graphs, their corresponding \matousek{}-type USOs, and their adjacency matrices.

\begin{figure}
    \centering
    \includegraphics{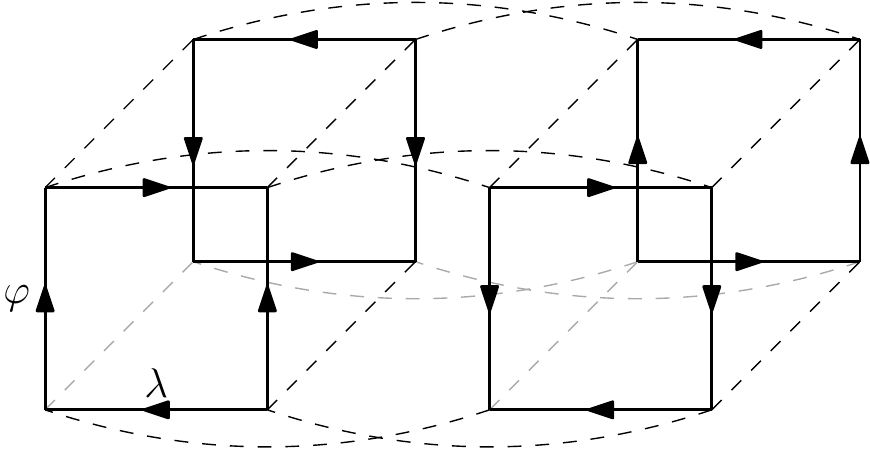}
    \caption{All $2$-faces spanned by $\lambda$ and $\phi$ in this $4$-dimensional \matousek{}-type USO have the same structure: $\phi$ \emph{influences} $\lambda$, but not the other way around.}
    \label{fig:dimensioninfluence}
\end{figure}

As the dimension influence graph $G$ is acyclic (apart from the loops), there must be a sink in every induced subgraph of $G$. Thus, in each face there is a dimension which is not influenced by any other dimension, and all edges of that dimension point in the same direction. We therefore say that each face is \emph{combed}, and that the whole \matousek{}-type USO is \emph{decomposable}.

\begin{figure}
    \centering
    \includegraphics[]{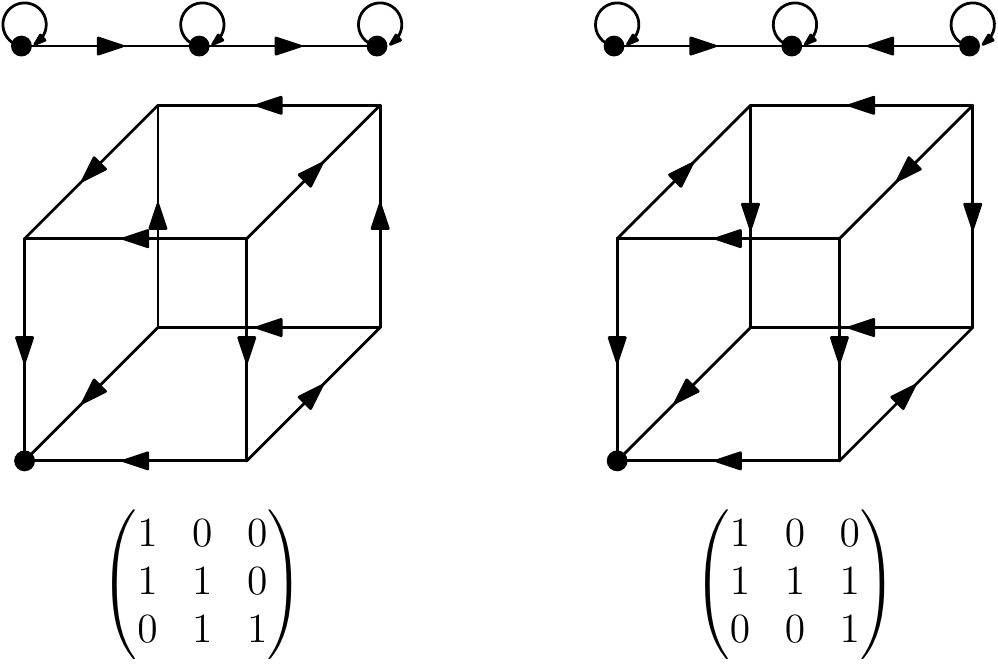}
    \caption{Top: The two forbidden graphs which cannot occur as induced subgraphs in dimension influence graphs of realizable \matousek{}-type USOs. Middle: The corresponding \matousek{}-type USOs with the sink at the bottom left vertex. Bottom: The adjacency matrices of the forbidden graphs. }
    \label{fig:forbidden}
\end{figure}

Weber and Gärtner showed that a \matousek{}-type USO is realizable if and only if its dimension influence graph does not contain one of the two graphs in \Cref{fig:forbidden} as an induced subgraph~\cite{weber2021matousek}. This directly implies the following characterization.
\begin{lemma}[{\cite[Theorem 4.5]{weber2021matousek}}]\label{lem:realizablegraphs}
A graph is the dimension influence graph of a realizable \matousek{}-type USO if and only if it is the reflexive transitive closure of a branching\footnote{A branching is a forest of rooted trees, where all edges are directed away from the roots.}.
\end{lemma}
See \Cref{fig:arborescences} for an example graph that is the dimension influence graph of a realizable \matousek{}-type USO.

\begin{figure}
    \centering
    \includegraphics{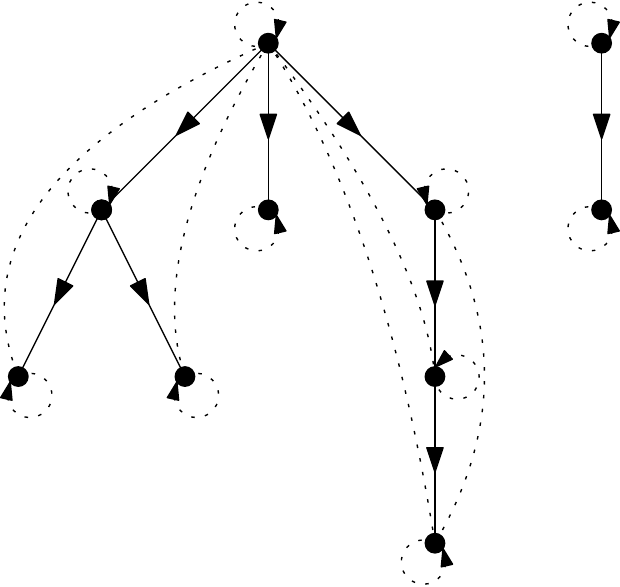}
    \caption{A branching with the edges added by taking the reflexive transitive closure (dashed).}
    \label{fig:arborescences}
\end{figure}

We can easily count the number of (realizable) \matousek{}-type USOs using classical graph-theoretic results. For general \matousek{}-type USOs, the dimension influence graph can be the reflexive closure of any directed acyclic graph.
\begin{lemma}
There are $2^{\Theta(n^2)}$ \matousek{}-type USOs.
\end{lemma}
\begin{proof}
Each pair of dimension influence graph and sink location describes a distinct \matousek{}-type USO. There are $2^n$ sink locations, and there are $2^{\Theta(n^2)}$ labelled directed acyclic graphs on $n$ vertices~\cite{stanley1973numberofdags}. Thus there are $2^n\cdot 2^{\Theta(n^2)}\in 2^{\Theta(n^2)}$ \matousek{}-type USOs.
\end{proof}

Similarly, a realizable \matousek{}-type USO can be described by the branching underlying the dimension influence graph, and the sink location.
\begin{lemma}
There are $2^n\cdot (n+1)^{n-1}\in 2^{\Theta(n\log n)}$ realizable \matousek{}-type USOs.
\end{lemma}
\begin{proof}
By Cayley's formula~\cite{cayley1889numberofarborescences}, there are $(n+1)^{n-1}$ labelled rooted forests on $n$ vertices. By directing all edges away from the roots, there is a bijection from the set of labelled rooted forests to the set of labelled branchings. All branchings have distinct reflexive transitive closures, thus there are $(n+1)^{n-1}$ valid dimension influence graphs for realizable \matousek{}-type USOs.
\end{proof}

As each vertex evaluation only provides $n$ bits of information, we get a lower bound on the number of vertex evaluations required to \emph{distinguish} \matousek{}-type USOs, but sink-finding can of course be easier.
\begin{corollary}\label{cor:informationlowerbound}
It takes at least $\Omega(n)$ vertex evaluations to distinguish all \matousek{}-type USOs, and at least $\Omega(\log n)$ vertex evaluations to distinguish all realizable \matousek{}-type USOs.
\end{corollary}

\subsection{Equivalence of Sink-Finding and Solving \texorpdfstring{$Mx=y$}{Mx=y}}
By simply rearranging some terms in \Cref{def:matousekuso}, we get the following observation.
\begin{observation}
\label{obs:equivOfQueries}
An orientation $u$ is a \matousek{}-type USO with dimension influence graph $G$ given by the adjacency matrix $M$ if and only if it fulfills
\[\forall x,y\in \{0,1\}^n: u(x)\xor u(y) = M(x\xor y).\]
\end{observation}

This looks quite innocent, and maybe not too useful, but we can use this equation to formulate an algebraic problem which is equivalent to sink-finding in \matousek{}-type USOs.

\begin{definition}[$Mx=y$ Problem]
For a matrix $M\in \{0,1\}^{n\times n}$ and a vector $y\in \{0,1\}^n$, the associated $Mx=y$ problem is to find the vector $x\in\{0,1\}^n$ fulfilling $Mx=y$. $M$ is not explicitly included as part of the problem instance, but only an oracle is provided to the algorithm. The oracle answers \emph{matrix-vector queries}: for any query $q\in\{0,1\}^n$, it returns $Mq$.
\end{definition}

Using \Cref{obs:equivOfQueries}, we can now show that up to a single additional query, the query complexities of the sink-finding problem and the $Mx=y$ problem are the same. We formalize this in the following theorem.

\begin{theorem}\label{thm:equivOfAlgorithms}
There exists a deterministic algorithm $\mathcal{A}$ to find the sink in a subclass $\mathcal{U}$ of \matousek{}-type USOs closed under reorientations in $f(n)$ vertex evaluations if and only if there exists a deterministic algorithm $\mathcal{B}$ to find $x$ fulfilling $Mx=y$ where $M$ can be the adjacency matrix of the dimension influence graph of any USO~$u\in\mathcal{U}$ in $f(n)-1$ matrix-vector queries.
\end{theorem}
\begin{proof}
We first prove the ``if'' direction. Given an algorithm $\mathcal{B}$ for solving $Mx=y$ in $f(n)-1$ queries we construct a sink-finding algorithm $\mathcal{A}$. The algorithm $\mathcal{A}$ first chooses an arbitrary vertex of the USO $u$, say $\zerovec$, and queries it to receive its outmap $u(\zerovec)$. It sets $y:=u(\zerovec)$ and sets $M$ to the (still unknown) adjacency matrix of the dimension influence graph of $u$. If $\mathcal{B}$ makes a query $q\in\{0,1\}^n$, $\mathcal{A}$ can simulate the matrix-vector oracle and answer this query by querying the vertex $q$ in the USO vertex evaluation oracle. The reply $Mq$ can be computed as $u(\zerovec)\xor u(q)$ by \Cref{obs:equivOfQueries}. Once $\mathcal{B}$ has found $x$ fulfilling $Mx=y$ (in at most $f(n)-1$ queries by assumption), $\mathcal{A}$ knows that $x$ must be the sink, as 
\[y=Mx=u(\zerovec)\xor u(x)=y\xor u(x) \text{ and thus } u(x)=\zerovec.\]
For each query of $\mathcal{B}$, $\mathcal{A}$ had to perform one query, with an additional query at the beginning to determine $y=u(\zerovec)$. In conclusion, $\mathcal{A}$ required at most $f(n)-1+1=f(n)$ queries.

Next, we prove the ``only if'' direction. Given an algorithm $\mathcal{A}$ for finding the sink in a USO $u$ from $\mathcal{U}$ in $f(n)$ queries, we construct an algorithm $\mathcal{B}$ to solve $Mx=y$. When $\mathcal{A}$ makes its first query, vertex $v_0$, $\mathcal{B}$ simulates the vertex evaluation oracle and answers with $u(v_0):=y$ from its given instance. Whenever $\mathcal{A}$ makes a query $v$, $\mathcal{B}$ computes the outmap $u(v)$ using
\[u(v):=u(v_0)\xor (M (v\xor v_0))=y\xor (M (v\xor v_0)).\]
By \Cref{obs:equivOfQueries}, $u$ is thus a \matousek{}-type USO with dimension influence graph with adjacency matrix $M$ and with $u(v_0)=y$. As $M$ is the dimension influence graph of a USO in $\mathcal{U}$, and $\mathcal{U}$ is closed under reorientations, $u$ must be in $\mathcal{U}$. Once $\mathcal{A}$ has found the sink $s$ (in at most $f(n)$ queries by assumption) $\mathcal{B}$ can compute the solution $x:=s\xor v_0$ to the system $Mx=y$, as
\[M(s\xor v_0)\overset{\text{Obs.\ref{obs:equivOfQueries}}}{=}u(s)\xor u(v_0)=\zerovec\xor u(v_0)=y.\]
For each query performed by $\mathcal{A}$ apart from the first one, $\mathcal{B}$ had to perform exactly one query. In conclusion, $\mathcal{B}$ required at most $f(n)-1$ queries.
\end{proof}

As the mapping $x\mapsto Mx$ computed by the oracle is linear and invertible, we can make the following two observations about any algorithm solving the $Mx=y$ problem.
\begin{observation}\label{obs:independentqueries}
If $x'$ is a linear combination of previously asked queries $x^{(1)},\ldots,x^{(k)}$, the response $Mx'$ can be computed without querying the oracle. Thus, every optimal algorithm only emits linearly independent queries.
\end{observation}

\begin{observation}\label{obs:span}
If $y$ is a linear combination of previously given replies $y^{(1)},\ldots,y^{(k)}$, i.e., if $y\in\Span (y^{(k)},\ldots,y^{(k)})$, the algorithm can find the solution $x$ with no additional queries.
\end{observation}

In some places it will be useful to interpret the matrix-vector queries in terms of vertex sets of the dimension influence graph.

\begin{observation}\label{obs:graphQueries}
For a query $q\in\{0,1\}^n$, the vertex set $\{i\in [n]:(Mq)_i=1\}$ contains exactly the vertices which have an odd number of in-neighbors among the vertex set $\{i\in [n] : q_i=1\}$ in $G$.
\end{observation}

\section{The General Case}\label{sec:generalcase}
It is not difficult to derive linear-time algorithms to find the sink in \matousek{}-type USOs.
\begin{theorem}\label{thm:jumpantipodal}
There exists an algorithm that finds the sink of an $n$-dimensional \matousek{}-type USO in $n$ vertex evaluations.
\end{theorem}
\begin{proof}
The \textsc{JumpAntipodal} algorithm begins at some arbitrary vertex $v:=v_0$. In each step, it queries~$v$, then jumps to the vertex $v\xor o(v)$. As \matousek{}-type USOs are decomposable, the USO is combed in some dimension $d$. After the first jump, the algorithm arrives at a vertex $v'$ in the facet towards which all edges of dimension $d$ are pointed. This facet can never be left again, and it is itself a decomposable USO of dimension~$n-1$. Applying this argument recursively, after at most~$n$ jumps \textsc{JumpAntipodal} has reached a $0$-dimensional USO --- a single vertex that must be the sink. It does not need to query this vertex anymore, and thus requires at most $n$ vertex evaluations.
\end{proof}

As can be seen from this proof, \textsc{JumpAntipodal} requires only $n$ vertex evaluations on all decomposable USOs, not only on the \matousek{}-type USOs. In contrast, it has been shown that even on some acyclic USOs it requires exponentially many queries~\cite{schurr2004phd}. Note that \textsc{JumpAntipodal} only finds the sink. If we are interested in recovering the whole structure of a \matousek{}-type USO, this can be easily achieved in $n$ matrix-vector queries, or $n+1$ vertex evaluations.

We are now going to prove \Cref{thm:lowerbound}, the matching lower bound to \Cref{thm:jumpantipodal}. The proof works in the framework of the $Mx=y$ problem, and is quite algebraic in nature.
\Cref{alg:adversarial} shows a strategy for an adversary to adaptively construct the matrix $M$ in a way to force every deterministic algorithm to use at least $n-1$ queries to find $x$. In terms of the dimension influence graph, this strategy can be seen as picking a sink $j$ in every iteration, and adding/removing some edges towards $j$. This ensures that the graph represented by $M$ always remains acyclic apart from the loops at each vertex.

\begin{algorithm}[h!]
\caption{Adversarial Construction}\label{alg:adversarial}
\begin{algorithmic}[1]
\State $M^{(0)}\gets I$
\For{$k\in\{1,\ldots,n-1\}$}
    \State $x^{(k)}\gets$ new linearly independent query from algorithm
    \If{$y\in\Span (M^{(k-1)}x^{(1)},\ldots,M^{(k-1)}x^{(k)})$}
        \State $X\gets\begin{pmatrix}x^{(1)} & \cdots & x^{(k)}
        \end{pmatrix}^T$
        \State $freevars\gets $ free variables of linear system of equations determined by $X$
        \State Pick $z^{(k)}$ such that $X z^{(k)}=e_k$ and $z^{(k)}_i=0$ for all $i\in freevars$
        \State Pick $j\in freevars$ such that $e_j$ is an eigenvector of $M^{(k-1)}$ \Comment{$j$ is a sink}
        \State $M^{(k)}\gets M^{(k-1)} + e_j {z^{(k)}}^T$ \Comment{Add $z^{(k)}$ to $j$-th row of $M^{(k-1)}$}
    \Else
        \State $M^{(k)}\gets M^{(k-1)}$ \Comment{No need to change $M$}
    \EndIf
    \State Answer query with $y^{(k)}:=M^{(k)}x^{(k)}$
\EndFor
\end{algorithmic}
\end{algorithm}

\lowerbound*
\begin{proof}
We show that the adversarial construction in \Cref{alg:adversarial} ensures that no algorithm can find the solution to the $Mx=y$ problem in fewer than $n-1$ queries. To prove this, we first show four auxiliary properties of the instances constructed by \Cref{alg:adversarial}: \begin{itemize}
    \item Feasibility: We can always pick $z^{(k)}$ and $j$ on lines~7 and~8 as defined.
    \item Consistency: The replies to previous queries remain consistent.
    \item Legality: The graph defined by $M$ remains a legal dimension influence graph, i.e., it is acyclic with added loops at every vertex.
    \item Uncertainty: After $k<n-1$ queries, the algorithm cannot yet distinguish between some instances with different solutions.
\end{itemize}

\textbf{Feasibility:} By \Cref{obs:independentqueries}, we can assume all queries $x^{(i)}$ to be linearly independent. Thus, the $k\times n$-dimensional matrix $X$ has rank $k<n$ and is underdetermined. We can thus set all free variables of $z^{(k)}$ given $Xz^{(k)}=e_k$ to be zero, and get a unique $z^{(k)}$ to be picked at line~7.

Whenever an additional linearly independent row is added to $X$, exactly one variable is removed from the set of free variables. Thus, after $k$ queries $n-k$ free variables remain. These were also free variables in all previous iterations, and therefore all vectors $z^{(k')}$ of iterations $k'\leq k$ have a $0$ at these coordinates. Therefore, for any variable $j\in freevars$, it holds that the column $j$ of $M^{(k-1)}$ must be equal to $e_j$, and thus $e_j$ must be an eigenvector of $M^{(k-1)}$. We conclude that any $j\in freevars$ can be picked on line~8.

\textbf{Consistency:} We prove that the possible change to $M^{(k-1)}$ in iteration $k$ has no effect on any query $x^{(k')}$ for $k'<k$. Note that if $M$ is changed in iteration $k$, we have $M^{(k)}=M^{(k-1)}+e_j {z^{(k)}}^T$, and thus $M^{(k)}x^{(k')}=M^{(k-1)}x^{(k')}+e_j {z^{(k)}}^T x^{(k')}$. As $z^{(k)}$ was picked such that $X z^{(k)}=e_k$, we have in particular ${x^{(k')}}^T z^{(k)}=0$, and thus
\[M^{(k)}x^{(k')}=M^{(k-1)}x^{(k')}+e_j {z^{(k)}}^T x^{(k')} = M^{(k-1)}x^{(k')}+e_j 0=M^{(k-1)}x^{(k')}.\]

\textbf{Legality:} As we start with $M^{(0)}=I$, we start with a loop at every vertex. As $z^{(k)}$ is added to the $j$-th row, and $z^{(k)}_j=0$, these loops are never removed. As $j$ is picked such that $e_j$ is an eigenvector of $M^{(k-1)}$, the $j$-th column of $M^{(k-1)}$ must be equal to $e_j$. This corresponds to the vertex~$j$ having no outgoing edges apart from the loop, i.e., $j$ is a sink. Changing the $j$-th row of~$M^{(k-1)}$ only adds or removes edges pointing \emph{towards} $j$. As $j$ is a sink, this cannot introduce any cycles. The graph described by $M^{(k)}$ thus remains a legal dimension influence graph.

\textbf{Uncertainty:} 
We first show that after each iteration, the algorithm cannot deduce the solution through linear combination, i.e.,
\begin{equation}
    \forall 0\leq k \leq n-1:\;\;y\not\in\Span(M^{(k)}x^{(1)},\ldots, M^{(k)}x^{(k)}).\label{eqn:desired}
\end{equation}
We prove this by induction. For $k=0$, the statement is trivially true. Assuming it holds for some~$k-1<n-1$, we show that it also holds for $k$. If in the $k$-th iteration the condition at line~4 is false, the statement also trivially follows. Otherwise, we must have
\begin{align}
y&\not\in\Span (M^{(k-1)}x^{(1)},\ldots, M^{(k-1)}x^{(k-1)}), \text{ but}\label{eqn:uncertainty1} \\
y&\in\Span (M^{(k-1)}x^{(1)},\ldots, M^{(k-1)}x^{(k-1)}, M^{(k-1)}x^{(k)}).\label{eqn:uncertainty2} \\
\intertext{Since $j$ is picked as a free variable of $X$, $e_j\not\in\Span (x^{(1)},\ldots, x^{(k)})$, and as $e_j$ is an eigenvector of the (invertible) $M^{(k-1)}$, it must also hold that}
    e_j&\not\in \Span (M^{(k-1)}x^{(1)},\ldots, M^{(k-1)}x^{(k)}).\label{eqn:uncertainty3}\\
\intertext{
\Cref{eqn:uncertainty1,eqn:uncertainty2} show that $M^{(k-1)}x^{(k)}$ is a required element in the linear combination of~$y$. \Cref{eqn:uncertainty3} shows that $e_j$ cannot be expressed as a linear combination of the $M^{(k-1)}x^{(k')}$. Therefore, if we add $e_j$ to the required element $M^{(k-1)}x^{(k)}$, $y$ can not be in the span anymore, i.e.,}
    y&\not\in\Span (M^{(k-1)}x^{(k)}+e_j,M^{(k-1)}x^{(1)},\ldots,M^{(k-1)}x^{(k-1)}).
\end{align}
This is equivalent to the desired \Cref{eqn:desired} for $k$, as $M^{(k)}x^{(k)}=M^{(k-1)}x^{(k)}+e_j$, and as shown in paragraph ``Consistency'', $M^{(k)}x^{(k')}=M^{(k-1)}x^{(k')}$ for all $k'<k$.

We can now show that after $k<n-1$ queries, there exist two matrices which are both consistent with the given replies but have different solutions. The first such matrix is $M^{(k)}$. The second matrix is the matrix $M^{(k+1)}$ constructed by the adversary if it would be given the solution for $M^{(k)}$ as an additional linearly independent query~$x^{(k+1)}:={M^{(k)}}^{-1}y$. As proven in previous paragraphs,~$M^{(k+1)}$ is legal and consistent with $M^{(k)}$ on all queries $x^{(1)},\ldots,x^{(k)}$. \Cref{eqn:desired} implies that $x^{(k+1)}$ is not the solution to $M^{(k+1)}$, proving that the two indistinguishable matrices have different solutions.

\textbf{Conclusion:}
Given any $n-1$ queries, \Cref{alg:adversarial} produces a series of legal matrices $M^{(k)}$ (Feasibility + Legality) which are always consistent with the previously given replies (Consistency). The algorithm cannot know the solution in fewer than $n-1$ queries, as it cannot distinguish between matrices with different solutions (Uncertainty). By \Cref{thm:equivOfAlgorithms}, we conclude that no algorithm can find the sink of an $n$-dimensional \matousek{}-type USO in fewer than $n$ vertex evaluations in the worst case.
\end{proof}

\section{The Realizable Case}\label{sec:realizablecase}
In this section we prove our second main result, the upper bound for the realizable case.
\sublinear*

To prove \Cref{thm:sublinear}, we provide a concrete algorithm in the matrix-vector query model to recover the matrix~$M$. The algorithm makes heavy use of the structure of the graph $G$ described by this matrix, which has to be the reflexive transitive closure of a branching, as we are only dealing with realizable \matousek{}-type USOs (recall \Cref{lem:realizablegraphs}). Recall that we can view the matrix-vector queries also as sets of vertices of the dimension influence graph (\Cref{obs:graphQueries}). In a slight abuse of notation, we will sometimes use the name $v$ of a vector $v\in\{0,1\}^n$ to also denote the set~$\{i\in[n]:v_i=1\}$.

We first take a closer look at the structure of $G$.
The underlying branching (technically, the unique reflexive transitive reduction of $G$) can be decomposed into levels, where the roots are on level $0$, and the children of a vertex on level $\ell$ are on level $\ell+1$. In the reflexive transitive closure $G$, we can see that the in-degree of each vertex is equal to its level plus $1$ (due to the loops).

\begin{observation}\label{obs:levelsanddegrees}
A vertex $v$ on level $\ell$ has exactly $\ell+1$ incoming edges, and $v$ has exactly one in-neighbor on each level $\ell'\in\{0,\ldots, \ell\}$.
\end{observation}

\begin{definition}
For a vertex $v$ on level $\ell$ and some level $\ell'<\ell$, the \emph{$\ell'$-ancestor of $v$} is the unique in-neighbor of $v$ on level $\ell'$. The \emph{parent of $v$} is the $\ell-1$-ancestor of $v$. The maximum level of any vertex in $G$ is denoted by $\ell_{max}$.
\end{definition}

Our proposed algorithm works in two main phases. In the first phase, the \emph{levelling}, it determines the level of each vertex in $O(\log n)$ queries. In the second phase, we use a divide-and-conquer approach to partition the vertices and perform queries to find the edges within each partition simultaneously, requiring $O(\log^2 n)$ queries in total.

\begin{algorithm}[h!]
\caption{Levelling}\label{alg:levelling}
\begin{algorithmic}[1]
\State $lvl\gets$ array of $n$ zeroes\Comment{Stores the level for every vertex}
\State $q\gets \onesvec$
\For{$i\in\{0,\ldots,\lceil \log_2 n\rceil-1\}$}
    \State $r\gets (M q) \xor q$ \Comment{Issues $1$ query}
    \For{$v\in\{1,\ldots,n\}$}
        \If{$r_v=1$}
            \State $lvl[v] \gets lvl[v] + 2^{i}$
        \EndIf
    \EndFor
    \State{$q\gets q\bitwiseand r$}\Comment{Bit-wise ``and'' operation}
\EndFor
\State \textbf{return} $lvl$
\end{algorithmic}
\end{algorithm}

\begin{lemma}
\Cref{alg:levelling} correctly computes the level of each vertex, using $O(\log n)$ queries.
\end{lemma}
\begin{proof}
\Cref{alg:levelling} issues $O(\log n)$ queries, one per iteration of the loop at line~3.

To prove correctness, we show that in each iteration $i$, the vertices in the vector $r$ are exactly those on levels $\ell$ where $Bin(\ell)_i=1$. From this follows that the level of each vertex is correctly recovered on line~7, one bit at a time. We show this by induction on $i$.

Note that a vertex $v$ is in $r=(M q)\xor q$ if it has an odd number of non-self in-neighbors in $q$, i.e., in-neighbors in $q\setminus\{v\}$.

For $i=0$ as the base case of this induction, $q=\onesvec$ and $r$ therefore contains all vertices with an odd number of non-self in-neighbors. By \Cref{obs:levelsanddegrees}, these are exactly the vertices on odd levels, i.e., those on levels where the least significant bit is $1$.

For the induction step, assume that for some $i$, the statement holds for all iterations $i'\leq i$. Thanks to the bit-wise ``and'' on line~8, the queried vertices $q$ in iteration $i+1$ are the vertices on levels $\ell'$ with $Bin(\ell')_{i'}=1$ for all $i'\leq i$, i.e., the binary representation of $\ell'$ ends with at least $i$ ones. The vector $(M q)\xor q$ contains all vertices with an odd number of non-self in-neighbors among these queried vertices. By \Cref{obs:levelsanddegrees}, these are all vertices on levels $\ell$ with an odd number of queried levels strictly above, i.e., with $|\{\ell'<\ell: \forall i'\leq i, Bin(\ell')_{i'}=1 \}| =_2 1$. This holds exactly for the levels~$\ell$ with $Bin(\ell)_{i+1} = 1$, thus proving the claim.
\end{proof}

We now know how to compute the level of each vertex in $O(\log n)$ time. It remains to show that we can also determine the edges connecting each consecutive two levels, and thus recover the whole graph. 

We first give the intuition for a simple strategy that given a level $\ell$ finds the $\ell$-ancestor of all vertices on levels $\geq \ell+1$. We can perform $\lceil \log_2 n\rceil$ queries, where the $i$-th query $q^{(i)}$ contains all vertices $v$ such that $\{v\text{ on level }\ell \text{ with }Bin(v)_i=1\}$. If a vertex is ``hit'', i.e., it is contained in the $i$-th reply $Mq^{(i)}$, we know that its $\ell$-ancestor must be in $q^{(i)}$. As each binary representation uniquely determines an integer, after all $\lceil \log_2 n\rceil$ queries, the $\ell$-ancestor is found.

It would be too costly to use this procedure for all levels on their own, as there can be up to~$n$ levels. We thus make use of the following observation, which follows directly from the transitivity and reflexivity of the dimension influence graph, as illustrated in \Cref{fig:cancelling}:

\begin{observation}\label{obs:filtering}
Let $\ell$ be some level, $q$ be some query and $Mq$ the corresponding response. Furthermore, let $w$ be some vertex on level $\ell'>\ell$, with the $\ell$-ancestor of $w$ being $a$. We form the alternative query $q'$ with $q'_v=1 \Longleftrightarrow (q_v=1\wedge level(v)\geq \ell\wedge v\not=a)$ by removing from $q$ the vertex~$a$ as well as all vertices strictly above level $\ell$. It holds that
\[(Mq')_w=(Mq)_w\xor (Mq)_{a}.\]
\end{observation}

\begin{figure}
    \centering
    \includegraphics{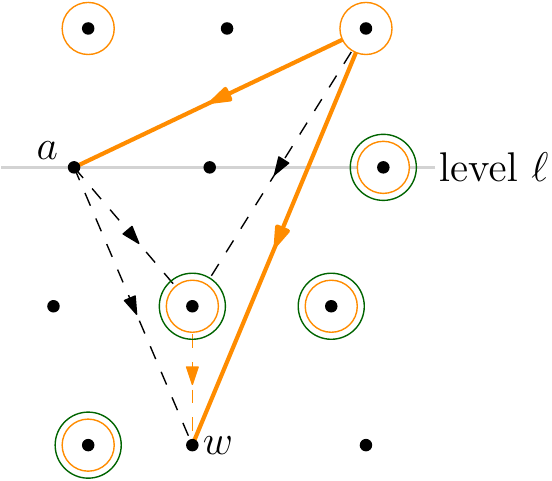}
    \caption{The orange-circled vertices are those included in the query $q$, the green-circled vertices are also in the query $q'$. All loops and all edges not induced by the vertices in $q$, $a$, and $w$ are omitted for legibility.
    The effect (bold) of the two vertices in $q$ but not in $q'$ on $a$ is the same as their effect on $w$, as $a$ is an ancestor of $w$ and the graph is transitive. This shows that~$(Mq')_w$ can be computed from $(Mq)_w$ and $(Mq)_a$.}
    \label{fig:cancelling}
\end{figure}

This observation allows us to ``filter out'' the effect of querying vertices above a certain level~$\ell$ on the vertices below~$\ell$, as long as their $\ell$-ancestors are known. Using this crucial tool, we can now solve the $Mx=y$ problem with a divide-and-conquer approach. If we use the previously given strategy to find all $\ell$-ancestors for a level $\ell$ roughly in the middle of the branching, we can then split the graph into two halves, the one above $\ell$ and the one below $\ell$. Using \Cref{obs:filtering}, we can proceed in each of the two subproblems simultaneously, as the effect of queries used to make progress in the upper half can be filtered out of the responses to the queries used to make progress in the lower half. We describe this process in \Cref{alg:divideandconquer}.

\begin{algorithm}[h!]
\caption{Divide-And-Conquer}\label{alg:divideandconquer}
\begin{algorithmic}[1]
\State{$subproblems \gets \{(0,\ell_{max})\}$}
\Statex{}\Comment{$subproblems$: a sorted set of disjoint intervals $(a_i,b_i)$ with $a_i<b_i<a_{i+1}$}
\State{$ancestor \gets$ $\ell_{max}\times n$-dimensional matrix of zeroes}
\Statex{}\Comment{$ancestor[\ell][v]$: contains the $\ell$-ancestor of vertex $v$, or $0$ if this is unknown/does not exist}
\While{$subproblems\not=\emptyset$}
    \State{$q\gets (0,\ldots,0)^T$}
    \For{$s\in\{0,\ldots,\lceil\log_2 n\rceil - 1\}$}\Comment{Perform binary search}
        \ForAll{$(a_i,b_i)\in subproblems$}
            \State{$m_i\gets \lfloor\frac{a_i+b_i}{2}\rfloor$}\Comment{Compute median level of subproblem}
            \ForAll{vertices $v$ on level $m_i$}
                \If{$Bin(v)_s=1$}
                    \State{$q_v \gets 1$}\Comment{Add $v$ to query}
                \EndIf
            \EndFor
        \EndFor
        \State{$r\gets M q$}\Comment{Issue a query}
        \ForAll{$(a_i,b_i)\in subproblems$}
            \ForAll{vertices $v$ on levels $>b_i$}
                \State{$r_v\gets r_v\xor r_{ancestor[b_i][v]}$}\Comment{Filter effect of level $m_i$ on levels $>b_i$}
            \EndFor
            \ForAll{vertices $v$ on levels $\{m_i+1,\ldots,b_i\}$}
                \If{$r_v=1$}\Comment{Detect ancestor}
                    \State{$ancestor[m_i][v]\gets ancestor[m_i][v] + 2^s$}
                \EndIf
            \EndFor
        \EndFor
    \EndFor
    \ForAll{vertices $v$, levels $\ell_1>\ell_2$}\Comment{``Transitivify'' $ancestor$}
        \If{$ancestor[\ell_1][v]\not=0$ and $ancestor[\ell_2][ancestor[\ell_1][v]]\not=0$}
            \State{$ancestor[\ell_2][v]\gets ancestor[\ell_2][ancestor[\ell_1][v]]$}
        \EndIf
    \EndFor
    \ForAll{$(a_i,b_i)\in subproblems$}\Comment{Split subproblems in half}
        \State{remove $(a_i,b_i)$ from $subproblems$}
        \If{$m_i>a_i$}
            \State{add $(a_i,m_i)$ to $subproblems$}
        \EndIf
        \If{$b_i>m_i+1$}
            \State{add $(m_i+1,b_i)$ to $subproblems$}
        \EndIf
    \EndFor
\EndWhile
\end{algorithmic}
\end{algorithm}

This algorithm keeps a list of subproblems, where each subproblem is described by an interval of levels $\{a_i,a_i+1,\ldots, b_i\}$, denoted by the tuple $(a_i,b_i)$. It is important that these subproblems are disjoint. In each iteration of the main loop, a median level $m_i$ of each subproblem is picked, and the previously mentioned procedure is performed to determine the $m_i$-ancestors of all vertices in the subproblem (lines~5--17). \Cref{obs:filtering} is applied with $\ell:=b_i$ to filter out the effect of the vertices queried for the subproblem $(a_i,b_i)$ from the later subproblems (lines~13 and~14). The transitive property of the ancestor relation is applied to make sure that the $m_i$-ancestors are also known for all vertices on levels below the considered subproblem~$(a_i,b_i)$ (lines~18--20). Finally, each subproblem is split into two parts, the one above $m_i$, and the one strictly below $m_i$. Subproblems consisting of a single level are ignored, as they contain no more edges to discover~(lines~21--26).

\newpage
\begin{lemma}
Knowing the level of each vertex, \Cref{alg:divideandconquer} correctly determines the parent of every vertex $v$ in $O(\log^2 n)$ queries.
\end{lemma}
\begin{proof}
Let us first assume that all $ancestor$ values required on line~14 are already known. By \Cref{obs:filtering} being applied on lines~13 and~14, the values $r_v$ read on line~16 are only influenced by the queried vertices on the level $m_i$ --- the effect of earlier subproblems is filtered out. Lines~4--17 therefore correctly determine the $m_i$-ancestors for all vertices on levels in $\{m_i+1,\ldots, b_i\}$. It only remains to prove that all $ancestor$ values required on line~14 are known, and that in the end of the algorithm, \emph{all} $ancestor[\ell][v]$ values are known.

We show the invariant that at the beginning of the main while loop, the values $ancestor[b_i][v]$ are known for all $v\in [n]$ and all $(a_i,b_i)\in subproblems$. At the beginning of the algorithm, this trivially holds, as there is only one subproblem $(0,\ell_{max})$, and no vertex has an $\ell_{max}$-ancestor. Whenever a subproblem $(a_i,b_i)$ is split into $(a_i,m_i)$ and $(m_i+1,b_i)$, there is only one new end of a subproblem, namely~$m_i$. On lines~4--17, the $m_i$-ancestor has just been computed for all vertices on levels in~$\{m_i+1,\ldots, b_i\}$, and the~$b_i$-ancestor was previously known for all vertices. On lines~18--20, this is combined to compute the $m_i$-ancestor for all vertices. The invariant thus holds.

To prove that all $ancestor[\ell][v]$ values are known at the end of the algorithm, we observe that every level is the end or the median of some subproblem at least once. In case the level is the end~$b_i$ of a subproblem, the $b_i$-ancestors are known by the previously proven invariant. In case the level is the median $m_i$ of a subproblem, the $m_i$-ancestors are computed in that iteration.

We conclude that \Cref{alg:divideandconquer} correctly determines all $ancestor[\ell][v]$ values, and thus also the parent of every vertex. In each iteration of the main while loop, $O(\log n)$ queries are issued, and as the size of the subproblems is halved in each iteration, there are at most $O(\log n)$ iterations. We conclude that \Cref{alg:divideandconquer} requires $O(\log^2 n)$ queries.
\end{proof}

\begin{proof}[Proof of \Cref{thm:sublinear}]
Using \Cref{alg:levelling} to compute the level of each vertex in the dimension influence graph $G$ and \Cref{alg:divideandconquer} to then compute all ancestors of all vertices, $G$ can be completely recovered in $O(\log^2 n)$ matrix-vector queries. Knowing $M$, $Mx=y$ can be solved with no additional queries using Gaussian elimination. By \Cref{thm:equivOfAlgorithms}, the sink of a realizable \matousek{}-type USO can be found in $O(\log^2 n)$ vertex evaluations.
\end{proof}

We believe that $\Theta(\log^2 n)$ is the best-possible number of queries to find the sink of a realizable \matousek{}-type USO. Due to the rigid structure of the dimension influence graph, a large portion of the $n$ bits of information in every reply from the oracle is redundant. We did not manage to prove this matching lower bound, but we can show a lower bound of $\Omega(\log n)$. Note that our algorithm solves the harder problem of recovering the whole structure of the USO, while the following lower bound holds for the easier problem of only finding the sink.

\begin{theorem}\label{thm:realizablelowerbound}
Every deterministic algorithm requires at least $\Omega(\log n)$ queries to find the sink of an $n$-dimensional realizable \matousek{}-type USO in the worst case.
\end{theorem}
\begin{proof}
We prove this in the matrix-vector query model. We set $y=\onesvec$ and begin with $M=I$. As~$y=\onesvec$, the desired $x$ such that $Mx=y$ is the vector corresponding to all roots of the dimension influence graph described by $M$, as their out-neighbor sets form a partition of all vertices.

During the construction, we enforce the invariant that the branching underlying the dimension influence graph is the union of disjoint paths. We call a path $p$ \emph{good}, if for each previously given query $q$, there is an even number of queried vertices within the path, i.e., $|q\cap p|=_2 0$.

As long as there is at least one good path $p_1$ and at least one other path $p_2$, all replies given to the algorithm are also consistent with the graph in which $p_2$ is attached to the end of $p_1$. In this alternative graph, the first vertex of $p_2$ is not a root. These two graphs thus have different solutions but are indistinguishable to the algorithm, and we conclude that the algorithm cannot know the set of roots at this point.

Note that at the beginning of the construction, when $M=I$ and no queries have arrived yet, the branching underlying the graph consists of $n$ paths of length $0$, which are all good.

Whenever the algorithm queries an odd number of vertices of some good paths, these paths are paired up. When two paths $p_1$ and $p_2$ are paired up, $p_2$ is attached to the end of $p_1$~(see \Cref{fig:goodpaths}). The query is then answered according to this new graph. As each previous query contained an even number of vertices in $p_1$ (as $p_1$ is good), the influence of these vertices onto the vertices in $p_2$ cancels out. The replies to these queries therefore remain consistent. As the newest query would have contained an odd number of vertices of both $p_1$ and $p_2$, the joined path is still good.

If there is an odd number of paths to be paired up, there is one leftover path. This path can not be paired up. It is no longer good, and will not be changed anymore.

We observe that when there are $k$ good paths before a query, there are at least $\lfloor k/2\rfloor$ good paths remaining after the query. We showed that the algorithm can only know the solution when there are no good paths left, or only a single path in total. As the graph contains $n$ good paths before the first query, it takes at least $\Omega(\log n)$ queries until the algorithm can know the solution.
\end{proof}

\begin{figure}
\begin{subfigure}{0.48\textwidth}
    \centering
    \includegraphics[page=1,scale=0.8]{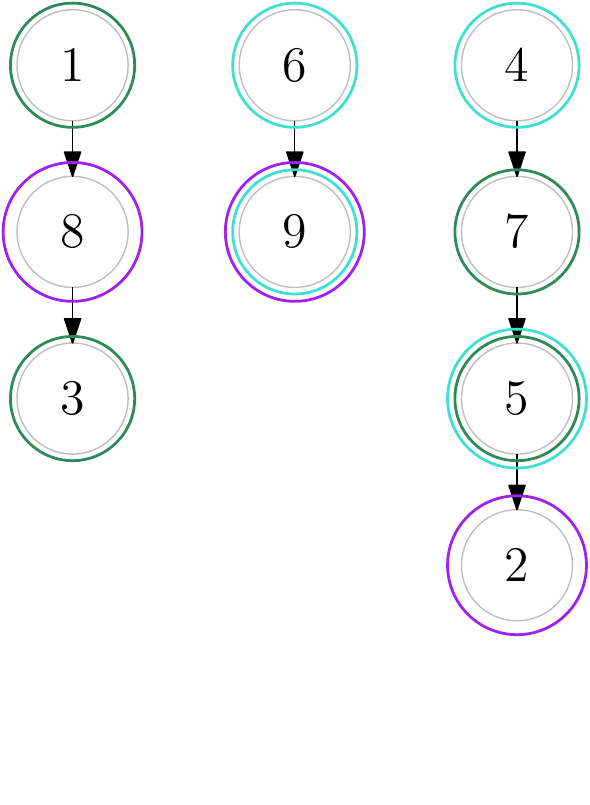}
    \caption{Before the purple query.}
\end{subfigure}
\begin{subfigure}{0.48\textwidth}
    \centering
    \includegraphics[page=2,scale=0.8]{figs/goodpaths.pdf}
    \caption{After the purple query.}
\end{subfigure}
\caption{After the green and cyan queries, all three paths in the left figure are good. To accommodate the purple query, the left and middle paths are paired up and joined. The combined path remains good. The right path is leftover, and is no longer good after the purple query. Loops and transitive edges are not shown.}\label{fig:goodpaths}
\end{figure}

\section{Conclusion}\label{sec:conclusion}
We have determined the query complexity of finding the sink in general \matousek{}-type USOs exactly. For realizable \matousek{}-type USOs, there remains an $O(\log n)$ gap between our lower and upper bound. While it would be interesting to close this gap, our main result --- the gap between the realizable and the general case --- is already well established by our bounds.

As the best-known sink-finding algorithms are randomized, it would be desirable to establish a complexity gap also for randomized algorithms. The upper bound naturally carries over, as all deterministic algorithms are also randomized algorithms. The most natural approach to establish a lower bound is applying Yao's principle~\cite{yao1977principle}: One needs to find a distribution over the \matousek{}-type USOs, such that every deterministic algorithm requires $\omega(\log^2 n)$ queries in expectation to find the sink of an USO drawn from this distribution. This then yields a lower bound for randomized algorithms. It is conceivable that a suitable distribution over the \matousek{}-type USOs created by our lower bound can achieve this, but we did not manage to prove any such result. On the flip side, it might also be possible to improve upon our algorithms both for the realizable and the general case using randomness, but we did not observe any straightforward benefit of randomness.

The most important open question implied by our results is whether there are other (larger and more practically relevant) USO classes which also admit such a complexity gap. Ultimately, we hope for such a gap to exist for the class of \emph{all} USOs. Considering the lack of a strong lower bound, this goal is still far away.

Finally, the connections between \matousek{}-type USOs and D-cubes can be examined further. Are (some of) the realizable \matousek{}-type USOs also D-cubes? Can our techniques used to find the sink of a \matousek{}-type USO be adapted to work for the less rigid D-cubes?

\newpage
\bibliographystyle{plainurl}
\bibliography{USO}

\end{document}